\documentclass[10pt,conference]{IEEEtran}

\pdfoutput=1

\usepackage[qm]{qcircuit}
\usepackage{mathtools}
\usepackage{float}
\usepackage[linesnumbered, ruled, vlined]{algorithm2e}
\usepackage{tikz}
\usetikzlibrary{trees, shadows, arrows.meta, shapes.multipart, backgrounds, calc, fit, positioning}
\usetikzlibrary{decorations.pathreplacing}
\usepackage{booktabs}
\usepackage{xcolor}
\usepackage{subcaption}
\usepackage{cite}
\usepackage{amsmath,amssymb,amsfonts}
\usepackage{hyperref}
\usepackage{soul}
\usepackage[capitalise,nameinlink]{cleveref}

\newtheorem{theorem}{Theorem}
\newtheorem{lemma}[theorem]{Lemma}
\newtheorem{proposition}[theorem]{Proposition}
\newtheorem{definition}[theorem]{Definition}
\newtheorem{corollary}[theorem]{Corollary}

\makeatletter
\@ifundefined{proof}{%
  \newenvironment{proof}[1][Proof]{\par\noindent\textit{#1.}\enspace}{\hfill$\square$\par\medskip}
}{}
\makeatother

\definecolor{OIgray}{RGB}{153,153,153}
\colorlet{neutral-left}{OIgray!15!white}
\definecolor{PastelIBMBlue}{RGB}{173,200,255}
\definecolor{PastelIBMPurple}{RGB}{188,180,255}
\definecolor{PastelIBMPink}{RGB}{244,182,204}
\definecolor{PastelIBMOrange}{RGB}{255,194,150}
\definecolor{PastelIBMYellow}{RGB}{255,225,150}
\definecolor{PastelIBMGreen}{RGB}{159,220,180}

\newcommand{\ket}[1]{\lvert #1 \rangle}

\newcommand{\floor}[1]{\lfloor #1 \rfloor}

\newcommand{\poly}{\mathrm{poly}}
\newcommand{\atantwo}{\mathrm{atan2}}

\DeclareMathOperator{\arcsinop}{arcsin}

\title{Efficient Complex-Valued State Preparation\\on Bucket Brigade QRAM}

\author{
\IEEEauthorblockN{Alessandro Berti}
\IEEEauthorblockA{Department of Computer Science and\\Department of Physics\\University of Pisa, Italy\\alessandro.berti@df.unipi.it}
\and
\IEEEauthorblockN{Francesco Ghisoni}
\IEEEauthorblockA{Department of Physics\\University of Pavia, Italy\\francesco.ghisoni01@universitadipavia.it}
}

\begin{document}

\maketitle

\begin{abstract}
Efficient quantum state preparation is a critical component in quantum algorithms that process large classical data, and it is fundamental to realizing quantum advantage in domains such as machine learning, quantum linear algebra, and quantum finance. Building on the framework of~\cite{berti2025efficient}, which integrates Bucket Brigade QRAM (BBQRAM) with a segment tree to achieve amplitude encoding in polylogarithmic query time, we present two improvements within the same architecture-aware framework. First, we remove the $U_{2\mathrm{CR}}$ subroutine by classically precomputing the rotation angles determined by the segment tree and storing these angles directly in the BBQRAM cells. The tradeoff is that the classically loaded QRAM stores precomputed fixed-point angles rather than raw subtree weights. Second, we extend the construction to complex-valued matrices $A \in \mathbb{C}^{M \times N}$ by storing a leaf phase alongside each precomputed rotation angle and using a two-step magnitude-then-phase procedure; the real signed case is naturally subsumed as a one-bit phase specialization. At unchanged $\mathcal{O}(\log_2^2(MN))$ BBQRAM query complexity, the QPU procedure reduces to BBQRAM retrievals and controlled-rotation cascades, with $\mathcal{O}(MN)$ memory cells per matrix and no reversible arithmetic on the QPU. 
\end{abstract}

\begin{IEEEkeywords}
Quantum State Preparation, Amplitude Encoding, Bucket Brigade QRAM, Segment Tree, Complex Amplitudes
\end{IEEEkeywords}

\section{Introduction}\label{sec:introduction}
Efficient quantum state preparation is fundamental to obtaining quantum advantages from quantum algorithms that process large data~\cite{hann2021practicality, luongo2021quantumalgorithms, prakash2014quantum}. Without a preparation procedure whose complexity scales at most polylogarithmically with the data size, the time required for encoding classical data into a quantum state can dominate the overall computational cost, resulting in a \emph{data bottleneck} that negates theoretical speedups of downstream quantum algorithms in machine learning~\cite{biamonte2017quantum, Rebentrost_2014_QSVM, kerenidis2020quantum, Wiebe_2012_data_fitting}, linear algebra~\cite{lloyd2010quantum, Wossnig_2018, bernasconiMatMul, guo2024quantumlinearalgebraneed}, quantum finance~\cite{rebentrost2022_finance, stochastic2022}, and quantum simulations of complex materials and molecules~\cite{Babbush_2018, Berry_2019, fomichev2024}. Given a matrix $A \in \mathbb{C}^{M \times N}$ with entries $a_{i,j}$, the goal of \emph{amplitude encoding} is to prepare the quantum state
\begin{equation}\label{eq:target_state}
\ket{A} \;=\; \frac{1}{\|A\|_F} \sum_{i=0}^{M-1}\sum_{j=0}^{N-1} a_{i,j}\,\ket{i}\ket{j},
\end{equation}
where the normalization factor coincides with the Frobenius norm $\|A\|_F = \sqrt{\sum_{i,j}|a_{i,j}|^2}$ of the matrix $A$, embedding the classical data directly into the amplitudes of a quantum state.

In~\cite{berti2025efficient}, the authors present an efficient state preparation algorithm that integrates the Bucket Brigade QRAM (BBQRAM)~\cite{giovannetti2008quantum, giovannetti2008architectures} with the Segment Tree data structure~\cite{cormen2022introduction}, achieving an explicit $\mathcal{O}(\log_2^2(MN))$ time bound under fault-tolerant assumptions. The key idea is to embed a segment tree of squared norms into the BBQRAM memory cells, and then traverse the tree level by level, using $U_{2\mathrm{CR}}$~\cite{chen_et_al:LIPIcs.ICALP.2023.38} at each level to convert basis-encoded segment-tree values into amplitudes. However, this approach presents two limitations that we address in the present work.

The first limitation concerns the $U_{2\mathrm{CR}}$ subroutine itself. To convert a pair of basis-encoded segment-tree values into a rotation angle, $U_{2\mathrm{CR}}$ requires four reversible arithmetic circuits, namely an adder, a divider, a square root, and an arcsine. Although these circuits have $\tilde{O}(1)$ theoretical cost under fixed precision, they represent the non-trivial source of complexity in the quantum procedure and introduce additional ancilla qubits. As noted in~\cite{berti2025efficient}, this makes $U_{2\mathrm{CR}}$ a non-negligible bottleneck on the path toward a practical fault-tolerant implementation of the algorithm.

The second limitation is the restriction to real-valued matrices $A \in \mathbb{R}^{M \times N}$. Many important applications require encoding complex-valued data: in quantum chemistry and quantum-simulation settings complex-valued data naturally arise~\cite{Babbush_2018}; and quantum linear algebra subroutines~\cite{Wossnig_2018} operate on complex matrices arising from discretizations of differential equations. In the original algorithm, the authors handle negative entries via a single sign bit and a phase-correction step at the end of the procedure; this mechanism encodes only the sign of real entries. Our goal is to give an explicit complex-valued construction within the same architecture-aware framework.

In this work, we address both limitations within the architecture-aware framework. The asymptotic BBQRAM query complexity remains the same as in~\cite{berti2025efficient}; the contribution is instead an architecture-aware simplification of the QPU procedure, together with an explicit complex-valued extension of the memory layout. More precisely, our contributions are:
\begin{enumerate}
\item \textbf{Elimination of $U_{2\mathrm{CR}}$ via classical precomputation.} We observe that the rotation angle applied by $U_{2\mathrm{CR}}$ at each level of the segment tree depends only on values known classically after preprocessing. We therefore precompute these angles classically and store them directly in the BBQRAM memory cells. Each amplitude-preparation iteration then reduces to a BBQRAM retrieval, a cascade of controlled-$R_y$ gates, and an uncomputation. This removes reversible addition, division, square root, and arcsine from the quantum procedure while preserving the $\mathcal{O}(\log_2^2(MN))$ BBQRAM query-time bound. In addition to this query cost, the online procedure contains $k=\log_2(MN)$ controlled-rotation cascades, contributing ${\mathcal{O}}(t\log_2(MN))$ controlled-rotation cost for fixed precision $t$. The improvement is therefore a QPU resource reduction, not a new asymptotic query bound.

\item \textbf{Extension to complex-valued matrices.} We give an explicit complex-amplitude construction within the architecture-aware framework of~\cite{berti2025efficient}. For $A \in \mathbb{C}^{M \times N}$, we decompose each entry in polar form $a_{i,j} = |a_{i,j}|\,e^{i\varphi_{i,j}}$ with $\varphi_{i,j} = \atantwo(\Im(a_{i,j}),\, \Re(a_{i,j}))$. The quantum procedure consists of two steps: magnitude preparation via cascades of controlled-$R_y$ gates, followed by phase encoding via one additional BBQRAM query and a cascade of controlled phase gates. We make the memory layout, zero-entry conventions, and fixed-point phase convention explicit. The total query time remains $\mathcal{O}(\log_2^2(MN))$.
\end{enumerate}

Our complexity statements separate three costs. The \emph{query time} counts the logical latency of coherent BBQRAM retrievals. The QPU resource count includes the address, work, and phase-target registers and the controlled rotations. The one-time classical construction and writing of the BBQRAM contents are counted separately and are intended to be amortized over repeated preparations of the same data.

We summarize the main result in the following informal theorem statement; the formal version appears in Theorem~\ref{thm:state_prep_complex}.
\begin{theorem}[Informal]\label{thm:main_informal}
Given $A \in \mathbb{C}^{M \times N}$ with $M = 2^m$, $N = 2^n$, and a fixed precision $t$, there exists a quantum algorithm that prepares
$\frac{1}{\|A\|_F}\sum_{i,j} a_{i,j}\ket{i}\ket{j}$
in BBQRAM query time $\mathcal{O}(\log_2^2 (MN))$, using $\Theta(\log_2 (MN))$ QPU qubits for fixed $t$, $\mathcal{O}(MN)$ BBQRAM memory cells of $2t$ bits each for complex data (or $t+1$ bits for real signed data), and $\mathcal{O}(MN)$ classical preprocessing time. No reversible arithmetic is required on the QPU.
\end{theorem}

We organize the manuscript as follows. In Section~\ref{sec:related_works}, we review related works on quantum state preparation, organized by the QRAM model each work assumes. In Section~\ref{sec:preliminaries}, we recall the necessary background on basis and amplitude encodings, the cascade of controlled rotations, the BBQRAM architecture, and the Segment Tree data structure. Sections~\ref{sec:eliminating_u2cr} and~\ref{sec:complex_extension} present our two main contributions, namely the elimination of $U_{2\mathrm{CR}}$ and the extension to complex-valued inputs, respectively. Eventually, Section~\ref{sec:conclusion} summarizes the main contributions and discusses directions for future research. A complete numerical walkthrough is available in Section~\ref{sec:numerical_example}.
\section{Related Work}\label{sec:related_works}

\begin{table}[t]
\centering
\caption{Comparison of quantum state preparation methods for $N$-dimensional vectors or $K = MN$-entry matrices. \textbf{Qubits} counts the qubits required on the QPU (output register and ancillas) for the preparation procedure. \textbf{Complex} indicates support for arbitrary complex amplitudes. \textbf{QRAM} denotes the QRAM model: \emph{oracle} treats it as an abstract black-box primitive, \emph{circuit} implements the Bucket Brigade architecture inside the quantum circuit, \emph{hardware} models the Bucket Brigade QRAM as a separate device, and \emph{no} indicates QRAM-free methods.}
\label{tab:methods}
\scriptsize
\begin{tabular}{lcccc}
\toprule
\textbf{Method} & \textbf{Time} & \textbf{Qubits} & \textbf{Complex} & \textbf{QRAM} \\
\midrule
\cite{Mottonen_2004} & $\mathcal{O}(2^n)$ & $\mathcal{O}(\log_2 N)$ & yes & no \\
\cite{kerenidis_et_al:LIPIcs.ITCS.2017.49} & $\mathcal{O}(\operatorname{polylog}(MN))$ & $\mathcal{O}(\log_2(MN))$ & no & oracle \\
\cite{prakash2014quantum} & $\tilde{\mathcal{O}}(\sqrt{\mathrm{nnz}(x)})$ & $\mathcal{O}(\log_2 N)$ & no & oracle \\
\cite{PhysRevA.110.032439} & $\mathcal{O}(s \log_2 N)$ & $\mathcal{O}(\log_2 N)$ & no & circuit \\
\cite{casares2020} & $\mathcal{O}(\operatorname{polylog}(MN))$ & $\mathcal{O}(\log_2(MN))$ & no & circuit \\
\cite{Zhang_2025} & $\mathcal{O}(\log_2 N)$ & $\mathcal{O}(N)$ & no & circuit \\
\cite{Clader_2022} & $\mathcal{O}(\log_2 N)$ & $\mathcal{O}(N)$ & no & circuit \\
\cite{berti2025efficient} & $\mathcal{O}(\log_2^2(MN))$ & $\mathcal{O}(\log_2(MN))$ & no & hardware \\
\addlinespace
\hline
\addlinespace
\textbf{This work} & $\mathcal{O}(\log_2^2(MN))$ & $\mathcal{O}(\log_2(MN))$ & \textbf{yes} & hardware \\
\bottomrule
\end{tabular}
\end{table}
This section reviews related works on quantum state preparation. We organize the discussion by the QRAM model each work assumes, progressing from abstract oracle treatments, through circuit-level implementations, to hardware-based architectures, and finally to QRAM-free approaches. Table~\ref{tab:methods} summarizes the comparison.

In~\cite{kerenidis_et_al:LIPIcs.ITCS.2017.49}, the authors introduce a data structure (the ``KP-tree'', equivalent to the standard segment tree, as observed in~\cite{berti2025efficient}) that enables amplitude encoding of a matrix $A \in \mathbb{R}^{M \times N}$ in $\mathcal{O}(\operatorname{polylog}(MN))$ time, treating each QRAM access as a unit-cost oracle. A separate phase oracle could in principle be added in such an abstract model; our claim is narrower, namely an explicit complex-valued layout and routing-compatible procedure in the hardware-BBQRAM model. In~\cite{prakash2014quantum}, the author extends this approach to sparse vectors by combining two oracle QRAM instances with a classical key--value map. Such approaches treat the QRAM as an abstract black-box primitive and do not account for routing overhead.

Several works move beyond the oracle abstraction and analyze the preparation procedure at the circuit level, treating the BBQRAM as a set of ancillary qubits that interact with the main register~\cite{PhysRevA.110.032439, PhysRevLett.129.230504}. In~\cite{casares2020}, the author adopts the Kerenidis--Prakash scheme within a circuit-level BBQRAM architecture but does not explicitly analyze the routing overhead; the overall time complexity remains $\mathcal{O}(\operatorname{polylog}(MN))$. In~\cite{Zhang_2025}, the authors encode the amplitudes of a vector using an explicit circuit-level BBQRAM model in $\mathcal{O}(\log_2 N)$ depth while requiring $\mathcal{O}(N)$ ancillary qubits. In~\cite{Clader_2022}, the authors study block encodings of matrices using a circuit-level BBQRAM with $\mathcal{O}(\log_2 N)$ depth, again at the cost of $\mathcal{O}(N)$ ancillas. We observe that, as discussed in~\cite{Jaques2025qramsurveycritique}, any fault-tolerant circuit-based BBQRAM requires $\Omega(2^n)$ total logical gates and $\Omega(\sqrt{2^n})$ non-Clifford gates, which undermines the polylogarithmic promise of the circuit model.

In contrast to the above approaches, in~\cite{berti2025efficient} the authors treat the QRAM as a \emph{separate hardware device} interacting with the quantum processing unit (QPU). The authors embed a segment tree of squared norms within the BBQRAM memory cells, prove an explicit $\mathcal{O}(\log_2^2(MN))$ time bound, and account for routing overhead via pipelined access. Their algorithm relies on $U_{2\mathrm{CR}}$~\cite{chen_et_al:LIPIcs.ICALP.2023.38} to convert pairs of basis-encoded segment-tree values into amplitude rotations, and targets only real-valued matrices. Our second contribution extends this particular architecture-aware framework to arbitrary complex phases.

For completeness, we also mention two families of methods that avoid QRAM entirely. Grover--Rudolph-style constructions~\cite{grover2002} and the uniformly controlled rotations of M\"{o}tt\"{o}nen et al.~\cite{Mottonen_2004} decompose an arbitrary $n$-qubit state into a tree of magnitude rotations followed by phase rotations, at the cost of $\mathcal{O}(2^n)$ gates for dense states. Variational approaches~\cite{Cerezo_2021, berti2024variational, PhysRevA.110.052615} parametrize the preparation circuit and optimize classically, but without general convergence guarantees. In contrast, our contribution provides a polylogarithmic-query implementation of the same amplitude-then-phase pattern within the architecture-aware framework of~\cite{berti2025efficient}, while simultaneously removing the reversible arithmetic of $U_{2\mathrm{CR}}$ from the QPU procedure. We therefore position the paper as a BBQRAM-query-model refinement rather than as a new dense-state gate-complexity bound.

\section{Preliminaries}\label{sec:preliminaries}
This section recalls the encodings, building blocks, and data structures used throughout the paper. We assume basic knowledge of quantum computing; for an introduction to the subject, we refer the reader to~\cite{nielsen2010quantum}. Throughout, a quantum register is denoted $\ket{\psi}_{\mathrm{name}}^{\mathrm{size}}$, where the subscript identifies the name of the register and the superscript indicates its size in qubits; size and name are omitted when clear from context. Section~\ref{sec:encodings} recalls the basis and amplitude encodings. Section~\ref{sec:cascade} introduces the cascade of controlled rotations, which we use extensively in the two contributions. Section~\ref{sec:bbqram} reviews the BBQRAM architecture with a focus on its retrieval cost, and Section~\ref{sec:segment_tree} recalls the Segment Tree data structure.

\subsection{Basis Encoding and Amplitude Encoding}\label{sec:encodings}
We first introduce two notions for encoding classical data into quantum computers: \emph{basis encoding} (Definition~\ref{def:basis_encoding}) and \emph{amplitude encoding} (Definition~\ref{def:amplitude_encoding}).
\begin{definition}[Basis encoding]\label{def:basis_encoding}
A $t$-bit string $x = x_{t-1}\dots x_1 x_0$ is represented as the computational basis state $\ket{x} = \ket{x_{t-1}\dots x_1 x_0}$, where the leftmost qubit is the most significant.
\end{definition}

\begin{definition}[Amplitude encoding]\label{def:amplitude_encoding}
Given a vector $x = [x_0, \dots, x_{N-1}] \in \mathbb{C}^N$ (padded to $N = 2^n$), its amplitude encoding is
$\ket{x} = \frac{1}{\|x\|_2}\sum_{j=0}^{N-1} x_j\,\ket{j}^n$.
\end{definition}

\noindent For a matrix $A \in \mathbb{C}^{M \times N}$ with $M = 2^m$ and $N = 2^n$, we adopt row-major indexing: $a_z = a_{i,j}$ with $z = i \cdot N + j$. We write $K = MN$ and $k = \log_2 K = m + n$. The amplitude encoding of $A$ reads
$\ket{A} = \frac{1}{\|A\|_F}\sum_{z=0}^{K-1} a_z\,\ket{z}^k$.

\subsection{Cascade of Controlled Rotations}\label{sec:cascade}
A central building block is the \emph{cascade of controlled rotation gates}, which decomposes a single-qubit rotation parametrized by an angle into a sequence of controlled rotations, each conditioned on one bit of the fixed-point angle register. Figure~\ref{fig:cascade} illustrates the magnitude cascade. We use two unsigned fixed-point conventions. Magnitude angles $\theta \in [0,\pi]$ are stored in a $t$-bit register representing multiples of $\Delta_\theta = 2^{2-t}$ in the interval $[0,4)$. Phases are stored modulo $2\pi$: for $\varphi \in (-\pi,\pi]$, the memory cell stores $\bar{\varphi} = \varphi \bmod 2\pi \in [0,2\pi)$ as a $t$-bit value representing multiples of $\Delta_\varphi = 2\pi/2^t$. Since $P(\varphi)$ is $2\pi$-periodic, using $\bar{\varphi}$ implements the same phase up to discretization error.

\begin{definition}[Cascade of controlled rotations]\label{def:cascade}
Let $a=a_{t-1}\cdots a_0$ be a $t$-bit fixed-point register, with $a_j$ the bit at position $j$ (so $a_{t-1}$ is most significant and $a_0$ least significant), representing $x=\sum_{j=0}^{t-1}a_j\,\Delta_j$, and let $b$ be a single target qubit. For any single-qubit rotation $R$ satisfying $R(\alpha+\beta)=R(\alpha)R(\beta)$, the cascade
\[
\mathrm{C}_{a}R_{\mapsto b}(x)
= \prod_{j=0}^{t-1}\mathrm{C}_{a_j}R_{\mapsto b}(\Delta_j)
\]
implements $R(x)$ on $b$. For magnitude angles we use $\Delta_j=2^{\,j+2-t}$, so the most significant bit $a_{t-1}$ controls $R(2)$ and the least significant bit $a_0$ controls $R(2^{2-t})$, matching Figure~\ref{fig:cascade}; for phases we use the corresponding modulo-$2\pi$ increments.
\end{definition}

\begin{lemma}[Cascade equivalence]\label{lem:cascade}
Under the conditions of Definition~\ref{def:cascade}, the cascade implements the rotation $R(\theta)$ on target qubit~$b$ with depth $\mathcal{O}(t)$. For a cascade of controlled $R_y$ gates:
$$\ket{\theta}_a^t\,\ket{0}_b \;\xmapsto{\;\text{Cascade } R_y\;}\; \ket{\theta}_a^t\bigl(\cos\frac{\theta}{2}\ket{0}_b + \sin\frac{\theta}{2}\ket{1}_b\bigr).$$
\end{lemma}

In this work, we use cascades of controlled $R_y$ gates for amplitude preparation and cascades of controlled $P$ gates for phase encoding.

\begin{figure}[t]
\centering
\resizebox{0.9\columnwidth}{!}{%
\Qcircuit @C=0.9em @R=1.6em  {
    & \qw & \ctrl{5}  & \qw     &  \qw   & \qw    & \qw   & \qw  & \ldots   &  & \qw     & \qw \\
    &\qw & \qw  & \qw      & \ctrl{4}  & \qw    & \qw     & \qw  & \ldots   &  & \qw    & \qw \\
    &\qw & \qw  & \qw      & \qw   & \qw    & \ctrl{3}   & \qw  & \ldots   &  & \qw      & \qw \\
    & \smash[b]{\vdots}   &     &          &      &        &      &      &   \smash[b]{\vdots}     &  &         &  \\
    &\qw & \qw   & \qw      & \qw   & \qw    & \qw    & \qw  & \ldots  &  & \ctrl{1}     & \qw
        \inputgroupv{1}{5}{0.8em}{4.5em}{\ket{\theta}^t_a} \\
    \lstick{\ket{0}_b}&\qw & \gate{R_y(2^1)} & \qw      & \gate{R_y(2^0)} & \qw    & \gate{R_y(2^{-1})}  & \qw  & \ldots   &  & \gate{R_y(2^{2-t})}& \qw
}%
}
\caption{Cascade of controlled rotations (Definition~\ref{def:cascade}). Each bit $a_i$ of the $t$-qubit register $\ket{\theta}^t_a$ controls one $R_y$ rotation on the target $\ket{0}_b$; their product reconstructs $R_y(\theta)$ via $R_y(\alpha+\beta)=R_y(\alpha)R_y(\beta)$.}
\label{fig:cascade}
\end{figure}
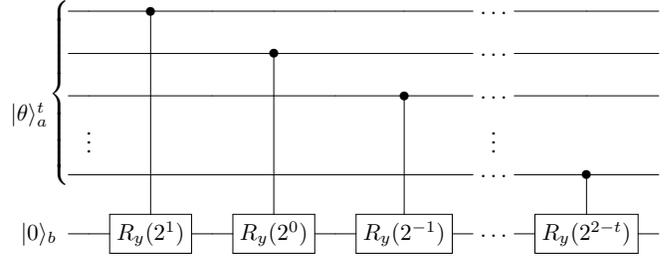

\subsection{Bucket Brigade QRAM}\label{sec:bbqram}
The BBQRAM is a physical device arranged as a binary tree of depth $n$ with $N = 2^n$ leaves (\emph{memory cells}) and $N-1$ internal nodes (\emph{switches}). Internal nodes act as routing switches that direct address qubits towards the target memory cell, while the leaves store the data values, which can be either classical or quantum bits. Figure~\ref{fig:bbqram} illustrates this architecture for $N = 8$ memory cells. The device supports queries of the form defined in Definition~\ref{def:bbqram_query}.
\begin{definition}[BBQRAM query]\label{def:bbqram_query}
A BBQRAM with $N = 2^n$ memory cells, each storing a $t$-qubit string $x_i$, implements
$$\mathsf{QRAM}\!:\; \ket{i}_{\mathrm{a}}^n\,\ket{b}_{\mathrm{d}}^t \;\longmapsto\; \ket{i}_{\mathrm{a}}^n\,\ket{b \oplus x_i}_{\mathrm{d}}^t,$$
where $\ket{i}^{n}_{\mathrm{a}}$ for $i \in \mathbb{N}$ denotes an address register of size $n$, thus indexing $N=2^n$ values, and $\ket{b \oplus x_i}^{t}_{\mathrm{d}}$ denotes the value register where $b$ is any bit string and $x_i \in \{0,1\}^t$, where $t \in \mathbb{N}$, represents the value associated with the address $\ket{i}_{\mathrm{a}}$
\end{definition}

\begin{figure}[t]
  \centering
  \resizebox{\columnwidth}{!}{%
    \begin{tikzpicture}[
        every label/.append style={font=\small},
        node/.style={
          draw, fill=white, circle, inner sep=0pt, font=\small,
          drop shadow, text width=1.5em, align=center
        },
        leaf/.style={
          shape=rectangle, rounded corners, draw, fill=white,
          inner sep=2pt, minimum width=3em, minimum height=1.5em,
          font=\small, align=center
        }
      ]
      \def\dx{2.8} \def\dy{0.8} \def\dxii{1.4} \def\leafsep{0.7}

      \node[node] (root) at (0,0) {$\ket{\bullet}$};
      \coordinate (above) at ($(root)+(0,\dy/1.5)$);
      \draw (above) -- (root);

      \node[node] (L)  at ($(root)+(-\dx,-\dy)$) {$\ket{\bullet}$};
      \node[node] (R)  at ($(root)+(\dx,-\dy)$)  {$\ket{\bullet}$};
      \node[node] (LL) at ($(L)+(-\dxii,-\dy)$) {$\ket{\bullet}$};
      \node[node] (LR) at ($(L)+(\dxii,-\dy)$)  {$\ket{\bullet}$};
      \node[node] (RL) at ($(R)+(-\dxii,-\dy)$) {$\ket{\bullet}$};
      \node[node] (RR) at ($(R)+(\dxii,-\dy)$)  {$\ket{\bullet}$};

      \node[leaf] (LLL) at ($(LL)+(-\leafsep,-\dy)$) {$\ket{x_0}$};
      \node[leaf] (LLR) at ($(LL)+(\leafsep,-\dy)$)  {$\ket{x_1}$};
      \node[leaf] (LRL) at ($(LR)+(-\leafsep,-\dy)$) {$\ket{x_2}$};
      \node[leaf] (LRR) at ($(LR)+(\leafsep,-\dy)$)  {$\ket{x_3}$};
      \node[leaf] (RLL) at ($(RL)+(-\leafsep,-\dy)$) {$\ket{x_4}$};
      \node[leaf] (RLR) at ($(RL)+(\leafsep,-\dy)$)  {$\ket{x_5}$};
      \node[leaf] (RRL) at ($(RR)+(-\leafsep,-\dy)$) {$\ket{x_6}$};
      \node[leaf] (RRR) at ($(RR)+(\leafsep,-\dy)$)  {$\ket{x_7}$};

      \draw (root) -- ++(-\dx,0) -- (L);  \draw (root) -- ++(\dx,0)  -- (R);
      \draw (L)  -- ++(-\dxii,0) -- (LL); \draw (L)  -- ++(\dxii,0) -- (LR);
      \draw (R)  -- ++(-\dxii,0) -- (RL); \draw (R)  -- ++(\dxii,0) -- (RR);
      \draw (LL) -- ++(-\leafsep,0) -- (LLL); \draw (LL) -- ++(\leafsep,0) -- (LLR);
      \draw (LR) -- ++(-\leafsep,0) -- (LRL); \draw (LR) -- ++(\leafsep,0) -- (LRR);
      \draw (RL) -- ++(-\leafsep,0) -- (RLL); \draw (RL) -- ++(\leafsep,0) -- (RLR);
      \draw (RR) -- ++(-\leafsep,0) -- (RRL); \draw (RR) -- ++(\leafsep,0) -- (RRR);
    \end{tikzpicture}
  }
  \caption{BBQRAM architecture for $N = 8$ memory cells. Circles are routing switches initialized in the wait state $\ket{\bullet}$; rectangles are memory cells, which can be either classical or quantum bits, storing $\ket{x_i}$.}
  \label{fig:bbqram}
\end{figure}
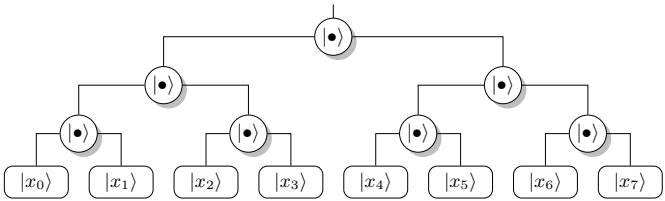

With \emph{pipelined routing}~\cite{hann2021resilience, xu2023systems}, a single retrieval costs $\mathcal{O}(\log_2 N)$ time, as formalized in Lemma~\ref{lem:bbqram_cost}.

\begin{lemma}[BBQRAM retrieval cost]\label{lem:bbqram_cost}
With \emph{pipelined routing}, the retrieval of a memory cell (or of a superposition of cells) from a BBQRAM with $N = 2^n$ memory cells costs $\mathcal{O}(\log_2 N)$ time.
\end{lemma}

All asymptotic time bounds in this paper are BBQRAM \emph{query-time} bounds in the hardware model of Lemma~\ref{lem:bbqram_cost}. The routing depth depends on the number of memory cells, while the cell width appears separately in the data-register size and memory footprint. Thus a $2t$-bit complex cell has the same asymptotic routing depth as a $t$-bit angle cell, but requires twice the data width and twice the memory bits. We count QPU qubits and gates separately from BBQRAM memory cells and switches. The one-time classical preprocessing and memory writing are outside the query time and are amortized over repeated preparations of the same data.

\subsection{Segment Tree}\label{sec:segment_tree}
The segment tree~\cite{de2008computational} is a complete binary tree where each node represents a segment of a vector. The leaves store individual values, and internal nodes store the result of an associative operation (in our case, summation) over their children. For $N$ values, the segment tree has $2N-1$ nodes, depth $\log_2 N$, and can be built in $\mathcal{O}(N)$ time. We note that we use the segment tree solely as a \emph{bottom-up sum aggregation tree}: each internal node stores the sum of its children's values, and we do not require range queries, lazy propagation, or dynamic updates. In the context of quantum state preparation, we construct a segment tree such that each leaf stores the squared modulus $|a_z|^2$ of a matrix entry, while each internal node stores the sum of squared moduli over its subtree, which represents the total probability weight of that branch. Definition~\ref{def:segment_tree} formalizes this construction.

\begin{definition}[Segment tree of squared moduli]\label{def:segment_tree}
For $A \in \mathbb{C}^{M \times N}$ with row-major indexing, $K = MN$, the segment tree $T$ is a complete binary tree of depth $k = \log_2 K$. A node at height $h \in [0, k]$, position $p \in [0, 2^h)$ stores
\[
T_{h,p} \;=\; \sum_{z = p \cdot K/2^h}^{(p+1)\cdot K/2^h - 1} |a_z|^2,
\]
with leaves $T_{k,p} = |a_p|^2$, root $T_{0,0} = \|A\|_F^2$, and parent--child relation $T_{h,p} = T_{h+1,2p} + T_{h+1,2p+1}$.
\end{definition}

Definition~\ref{def:segment_tree} generalizes the definition in~\cite{berti2025efficient} from $|a_z|^2$ with $a_z \in \mathbb{R}$ to $|a_z|^2 = \Re(a_z)^2 + \Im(a_z)^2$ for $a_z \in \mathbb{C}$. The structure of the tree is identical in both cases; only the computation of the leaf values changes. We also observe that the square root of the root value coincides with the Frobenius norm: $\sqrt{T_{0,0}} = \|A\|_F$. When the matrix dimensions are not powers of two, the matrix is zero-padded to the next power of two in each dimension; the zero entries contribute zero amplitude and do not affect the output state. We assume $A$ is not the all-zero matrix. If a subtree has zero total weight, the corresponding rotation angle may be set to $0$ by convention, since no amplitude ever reaches that branch.

\section{Eliminating $U_{2\mathrm{CR}}$ via Classical Precomputation}\label{sec:eliminating_u2cr}
In this section, we describe how to eliminate the $U_{2\mathrm{CR}}$ subroutine from the magnitude-preparation loop of~\cite{berti2025efficient}. The key point is that $U_{2\mathrm{CR}}$ computes a rotation angle from two segment-tree values that are already known after classical preprocessing. Rather than storing the two sibling subtree weights required by $U_{2\mathrm{CR}}$, we therefore construct a derived data structure --- the \emph{angle tree} $\Theta$ --- with one angle per amplitude-splitting step, halving the data width loaded by each magnitude-preparation query. We do \emph{not} try to encode arbitrary real signs into internal tree rotations: the internal nodes of $\Theta$ correspond to splittings between non-negative subtree sums and therefore carry only unsigned magnitude information. Sign or phase information is needed only at the leaf level, where individual matrix entries appear; we handle this via the sign bit of Corollary~\ref{cor:real_signed} for real signed data and via the phase layer of Definition~\ref{def:complex_segment_tree} for complex data. Section~\ref{sec:classical_precompute} introduces the classical precomputation of the rotation angles and the resulting angle tree $\Theta$. Section~\ref{sec:angle_memory_layout} specifies the BBQRAM memory layout for $\Theta$. Eventually, Section~\ref{sec:modified_alg} presents the magnitude-preparation procedure without reversible arithmetic and analyzes its complexity.

\subsection{Classical Precomputation of Rotation Angles}\label{sec:classical_precompute}
In the original algorithm~\cite{berti2025efficient}, at each level of the segment tree, the $U_{2\mathrm{CR}}$ subroutine reads the two sibling subtree weights from the BBQRAM, performs reversible addition, division, square root, and arcsine to compute the corresponding amplitude-splitting angle $\theta_z$, and then applies a cascade of controlled $R_y$ gates. We observe that the \emph{only output} of $U_{2\mathrm{CR}}$ that persists after uncomputation is the rotation applied to the target qubit. Since the segment-tree values are all known classically after preprocessing, the angle $\theta_z$ can be computed \emph{classically} and stored directly in the BBQRAM memory cell. Formally, let $l(z) = \floor{\log_2 z} + 1$ and $d(z) = z - 2^{\floor{\log_2 z}}$ denote the level and position functions from~\cite{berti2025efficient}.

\begin{figure}[t]
\centering
\resizebox{\columnwidth}{!}{%
\begin{tikzpicture}[
    >=stealth,
    every node/.style={align=center},
    angbase/.style={
      draw=black!60, line width=0.4pt,
      rounded corners=2pt,
      minimum width=1.2cm, minimum height=0.6cm,
      font=\scriptsize,
      drop shadow={shadow xshift=0.4pt, shadow yshift=-0.4pt, opacity=0.25}},
    anglevzero/.style={angbase, fill=PastelIBMPurple},
    anglevone/.style={angbase, fill=PastelIBMBlue},
    anglevtwo/.style={angbase, fill=PastelIBMPink},
    signbit/.style={
      draw=black!60, line width=0.4pt,
      rounded corners=2pt,
      minimum width=0.7cm, minimum height=0.6cm,
      fill=PastelIBMOrange,
      font=\scriptsize,
      drop shadow={shadow xshift=0.4pt, shadow yshift=-0.4pt, opacity=0.25}},
    treeedge/.style={draw=black!50, line width=0.5pt, ->}]
    \node[anglevzero] (r)   at ( 0  , 0)  {$\Theta_{0,0}$};
    \node[anglevone]  (a10) at (-2.4,-1)  {$\Theta_{1,0}$};
    \node[anglevone]  (a11) at ( 2.4,-1)  {$\Theta_{1,1}$};
    \node[anglevtwo]  (a20) at (-3.6,-2)  {$\Theta_{2,0}$};
    \node[anglevtwo]  (a21) at (-1.2,-2)  {$\Theta_{2,1}$};
    \node[anglevtwo]  (a22) at ( 1.2,-2)  {$\Theta_{2,2}$};
    \node[anglevtwo]  (a23) at ( 3.6,-2)  {$\Theta_{2,3}$};
    \draw[treeedge] (r)   -- (a10);  \draw[treeedge] (r)   -- (a11);
    \draw[treeedge] (a10) -- (a20);  \draw[treeedge] (a10) -- (a21);
    \draw[treeedge] (a11) -- (a22);  \draw[treeedge] (a11) -- (a23);
    \node[signbit] (s0) at (-4.2,-3.1) {$s_0$};
    \node[signbit] (s1) at (-3.0,-3.1) {$s_1$};
    \node[signbit] (s2) at (-1.8,-3.1) {$s_2$};
    \node[signbit] (s3) at (-0.6,-3.1) {$s_3$};
    \node[signbit] (s4) at ( 0.6,-3.1) {$s_4$};
    \node[signbit] (s5) at ( 1.8,-3.1) {$s_5$};
    \node[signbit] (s6) at ( 3.0,-3.1) {$s_6$};
    \node[signbit] (s7) at ( 4.2,-3.1) {$s_7$};
    \draw[treeedge] (a20) -- (s0);  \draw[treeedge] (a20) -- (s1);
    \draw[treeedge] (a21) -- (s2);  \draw[treeedge] (a21) -- (s3);
    \draw[treeedge] (a22) -- (s4);  \draw[treeedge] (a22) -- (s5);
    \draw[treeedge] (a23) -- (s6);  \draw[treeedge] (a23) -- (s7);
\end{tikzpicture}}
\caption{The angle tree $\Theta$ for $K=8$ entries (Definition~\ref{def:angle_segment_tree}); nodes hold one $t$-bit unsigned rotation angle, depth $\log_2 K - 1 = 2$, node count $K-1 = 7$. The orange leaf row shows the one-bit sign field $\{s_0,\dots,s_7\}$ added for real signed data (Corollary~\ref{cor:real_signed}); the complex case replaces this by the $t$-bit phase layer $\Phi$ of Figure~\ref{fig:gamma_data_structure}.}
\label{fig:angle_tree_theta}
\end{figure}
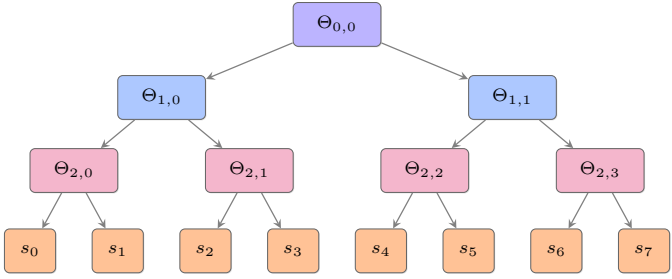
\begin{definition}[Precomputed rotation angle]\label{def:theta_precompute}
For each memory cell index $z \in \{1, \dots, K-1\}$, with $T_L(z) = T_{l(z),2d(z)}$ and $T_R(z) = T_{l(z),2d(z)+1}$, define
\[
\theta_z = 2\arcsinop\!\left(\sqrt{\dfrac{T_R(z)}{T_L(z) + T_R(z)}}\right)
\]
when $T_L(z)+T_R(z)>0$, and set $\theta_z=0$ otherwise.
\end{definition}
The angle $\theta_z$ has the key property that
\begin{equation}\label{eq:ry_theta}
R_y(\theta_z)\ket{0} = \sqrt{\frac{T_L(z)}{T_L(z) {+} T_R(z)}}\,\ket{0} + \sqrt{\frac{T_R(z)}{T_L(z) {+} T_R(z)}}\,\ket{1}.
\end{equation}
Therefore, retrieving the precomputed $\theta_z$ and applying the cascade of $R_y$ gates produces, up to the fixed-point precision $\varepsilon = \mathcal{O}(2^{-t})$, the same amplitude encoding as $U_{2\mathrm{CR}}$.

The collection of precomputed angles forms a derived data structure with one node per amplitude-splitting step, which we make explicit in Definition~\ref{def:angle_segment_tree}.
\begin{definition}[Angle tree]\label{def:angle_segment_tree}
Let $T$ be the segment tree of squared moduli of Definition~\ref{def:segment_tree}. The associated \emph{angle tree} $\Theta$ is a complete binary tree of depth $\log_2 K - 1$ with one node per internal sibling-pair splitting in $T$. For $h \in [0, \log_2 K - 1]$ and $p \in [0, 2^h)$,
\[
\Theta_{h,p} \;=\; \theta_{z}, \qquad z \;=\; 2^h + p,
\]
where $\theta_z$ is the precomputed rotation angle of Definition~\ref{def:theta_precompute} that splits the amplitude reaching $T_{h,p}$ between its two children $T_{h+1,2p}$ and $T_{h+1,2p+1}$. Thus $\Theta$ has $K-1$ nodes, can be built from $T$ in $\mathcal{O}(K\poly(t))$ classical time, and stores one $t$-bit fixed-point angle per node.
\end{definition}

Figure~\ref{fig:angle_tree_theta} shows the structure of $\Theta$ for $K=8$, together with the one-bit sign field $\{s_0,\dots,s_7\}$ attached at the leaves for real signed data (Corollary~\ref{cor:real_signed}). Each node $\Theta_{h,p}$ corresponds to one amplitude-splitting step; traversing $\Theta$ from root to leaves mirrors the $k$ iterations of the magnitude-preparation loop, with the level-$h$ angles loaded into the working register at iteration $h+1$. Compared with the original $T$, the angle tree has one fewer level (no leaves storing $|a_z|^2$), reflecting that the final amplitude split is the only place where individual matrix entries appear; the corresponding sign or phase information is therefore stored in a separate leaf field disjoint from $\Theta$ --- one bit per leaf for real signed data, $t$ bits per leaf for complex data (the phase layer $\Phi$ of Definition~\ref{def:complex_segment_tree}, depicted in Figure~\ref{fig:gamma_data_structure}).

\subsection{Memory Layout for the Angle Tree}\label{sec:angle_memory_layout}
We now specify how the angle tree is mapped into BBQRAM memory cells.

\begin{proposition}[Memory layout for the angle tree]\label{prop:angle_tree_mapping}
Let $\Theta$ be the angle tree associated with a nonzero $A\in\mathbb{C}^{M\times N}$, with $K=MN$ a power of two and $\Theta$ defined as in Definition~\ref{def:angle_segment_tree}. We map $\Theta$ into a BBQRAM with $K$ memory cells $\{\ket{L_z}\}_{z=0}^{K-1}$, where each cell contains a $t$-qubit angle field:
\begin{equation}
\label{eq:angle_tree_mapping}
\ket{L_z}
=
\begin{cases}
\ket{0}^{t}, & \textrm{if } z=0, \\[1ex]
\ket{\Theta_{l(z)-1,\,d(z)}}^{t}, & \textrm{otherwise,}
\end{cases}
\end{equation}
with
\[
l(z)=\lfloor\log_2 z\rfloor+1,
\qquad
d(z)=z-2^{\lfloor\log_2 z\rfloor}.
\]
Here $\ket{L_0}=\ket{0}^t$ is a dummy cell that fixes the uniform width; each cell $z\geq1$ stores exactly one fixed-point amplitude-splitting angle. These angles are unsigned magnitude angles. Real signs and complex phases are supplied by a separate leaf field, as in Corollary~\ref{cor:real_signed} and Definition~\ref{def:complex_segment_tree}.

For real signed data, each cell extends to $t+1$ bits by appending the one-bit sign field $s_z=\mathbf{1}[a_z<0]$:
\begin{equation}
\label{eq:real_signed_mapping}
\ket{L_z^{\mathrm{real}}}\;=\;\ket{L_z}^{t}\,\ket{s_z}^{1},
\qquad z\in\{0,1,\dots,K-1\},
\end{equation}
with $\ket{L_z}$ as in Equation~\ref{eq:angle_tree_mapping} and $s_0=\mathbf{1}[a_0<0]$ (the final sign query reaches the leaf branch $z=0$, so $s_0$ must carry the sign of $a_0$). The complex case replaces the one-bit sign by a $t$-bit phase, giving the $2t$-bit layout of Proposition~\ref{prop:memory_layout_complex}.
\end{proposition}

The BBQRAM cells may in principle store either classical or quantum bits; this changes only
the retrieval mechanism, not the logical query used here~\cite{Hann_2019}. In our construction the stored
values are classically precomputed fixed-point strings, so writing the memory is cheap:
one may initialize the cells by classical control, or equivalently by a parallel layer of
Pauli-$X$ gates on the qubits whose stored bits are equal to one.

\begin{algorithm}[t]
\caption{State preparation without $U_{2\mathrm{CR}}$ (real signed data)}\label{alg:magnitude}
\KwIn{$A \in \mathbb{R}^{M \times N}$, $K=MN$, $k=\log_2 K$, precision $t$}
\KwOut{$\frac{1}{\|A\|_F}\sum_z a_z\ket{z}$}
\BlankLine
\tcc{Step 1: Classical preprocessing}
Build segment tree $T$ on $\{|a_z|^2\}_{z=0}^{K-1}$\;
Write dummy $\theta_0=0$ and $s_0=\mathbf{1}[a_0<0]$ into BBQRAM cell $0$\;
\For{$z = 1$ \KwTo $K-1$}{
  $\theta_z \leftarrow 0$ if $T_L(z)+T_R(z)=0$, otherwise $2\arcsin\!\bigl(\sqrt{\frac{T_R(z)}{T_L(z)+T_R(z)}}\bigr)$\;
  $s_z \leftarrow \mathbf{1}[a_z<0]$\;
  Write $\ket{\theta_z}$ into the angle field and $\ket{s_z}$ into the sign field of BBQRAM cell $z$\;
}
\tcc{Step 2: Magnitude preparation}
Initialize $\ket{0}^t_{\mathrm{w}_\theta}\ket{0}_{\mathrm{w}_s}\ket{0}_{\mathrm{v}}\ket{0\cdots01}_{\mathrm{a}}$\;
\For{$h = 1$ \KwTo $k$}{
  Retrieve $\ket{\theta_z}\ket{s_z}$ into $\mathrm{w}_\theta,\mathrm{w}_s$ \tcp*{$\mathcal{O}(\log_2 K)$}
  Apply cascade $\mathrm{C}_{\mathrm{w}_\theta}R_{y\mapsto\mathrm{v}}(\theta_z)$ \tcp*{$\mathcal{O}(t)$}
  Uncompute $\mathrm{w}_\theta,\mathrm{w}_s$ \tcp*{$\mathcal{O}(\log_2 K)$}
  Left circular shift on $\mathrm{v},\mathrm{a}$\;
}
\tcc{Step 3: Sign encoding}
Retrieve $\ket{\theta_z}\ket{s_z}$ into $\mathrm{w}_\theta,\mathrm{w}_s$ \tcp*{$\mathcal{O}(\log_2 K)$}
Apply $\mathrm{C}_{\mathrm{w}_s}Z_{\mathrm{v}}$ \tcp*{$(-1)^{s_z}$ on the $\ket{1}_{\mathrm{v}}$ branch}
Uncompute $\mathrm{w}_\theta,\mathrm{w}_s$ \tcp*{$\mathcal{O}(\log_2 K)$}
\Return state on register $\mathrm{a}$\;
\end{algorithm}
\subsection{Magnitude Preparation without Reversible Arithmetic}\label{sec:modified_alg}
Assume that the angle tree is stored according to Proposition~\ref{prop:angle_tree_mapping}. Each BBQRAM query in Algorithm~\ref{alg:magnitude} loads the full $(t{+}1)$-bit cell into $(\mathrm{w}_\theta,\mathrm{w}_s)$; only the angle field is used in the magnitude rotation, while the sign field is uncomputed by the reverse query during Step~2 and consumed by the controlled-$Z$ in Step~3. This single field replaces the two $t$-bit sibling values loaded by the original implementation, and it exhausts the nodes of $\Theta$ exactly once as $z$ ranges over $\{1,\dots,K-1\}$. Algorithm~\ref{alg:magnitude} provides the pseudocode for the complete real signed procedure, including the final controlled-$Z$ sign-encoding step (Corollary~\ref{cor:real_signed}).

\begin{definition}[Circular-shift convention]\label{def:circular_shift}
Let the address register be ordered as $\mathrm{a}=a_{k-1}\cdots a_0$, with $a_{k-1}$ the most significant routing bit. The left circular shift on the ordered wires $(\mathrm{v},a_{k-1},\ldots,a_0)$ is the unitary
\[
\mathsf{Shift}\ket{b}_{\mathrm{v}}\ket{a_{k-1}\cdots a_0}_{\mathrm{a}}
=\ket{a_{k-1}}_{\mathrm{v}}\ket{a_{k-2}\cdots a_0 b}_{\mathrm{a}} .
\]
\end{definition}

\begin{lemma}[Routing-marker invariant]\label{lem:routing_marker}
Let $\mathrm{bin}_h(p)$ denote the $h$-bit binary representation of $p$. After the $h$-th magnitude-preparation iteration, including the shift of Definition~\ref{def:circular_shift}, the clean work-register state is
\[
\frac{1}{\|A\|_F}\sum_{p=0}^{2^h-1}\sqrt{T_{h,p}}\,
\ket{0}^{t}_{\mathrm{w}}\ket{0}_{\mathrm{v}}
\ket{0^{k-h-1}1\,\mathrm{bin}_h(p)}_{\mathrm{a}}
\]
for $1\leq h<k$, and
\[
\frac{1}{\|A\|_F}\sum_{p=0}^{K-1}\sqrt{T_{k,p}}\,
\ket{0}^{t}_{\mathrm{w}}\ket{1}_{\mathrm{v}}\ket{\mathrm{bin}_k(p)}_{\mathrm{a}}
\]
for $h=k$.
\end{lemma}

\begin{theorem}[Magnitude preparation without $U_{2\mathrm{CR}}$]\label{thm:magnitude_prep}
Given nonzero $A \in \mathbb{C}^{M \times N}$ with $M=2^m$, $N=2^n$, and $K=MN$, the precomputed-angle BBQRAM layout prepares
\[
\ket{|A|}=\frac{1}{\|A\|_F}\sum_{z=0}^{K-1}|a_z|\ket{z}^k
\]
in BBQRAM query time $\mathcal{O}(\log_2^2 K)$ using $k+t+1$ QPU qubits, $K$ memory cells with a $t$-bit angle field, and no reversible arithmetic on the QPU.
\end{theorem}

\begin{proof}
The level-by-level traversal, shift schedule, and routing-marker invariant are those of~\cite{berti2025efficient}. The only change is local: the $U_{2\mathrm{CR}}$ block is replaced by a BBQRAM retrieval of the precomputed angle $\Theta_{l(z)-1,d(z)}$ followed by the controlled-$R_y$ cascade of Equation~\ref{eq:ry_theta}. This applies exactly the rotation that $U_{2\mathrm{CR}}$ would have computed from the two child subtree weights, with the same zero-subtree convention. Lemma~\ref{lem:routing_marker} therefore yields, after $k$ iterations, the clean marker state $\ket{1}_{\mathrm{v}}$ and amplitudes $|a_z|/\|A\|_F$ on the leaf address register.

Each iteration uses one retrieval, one reverse retrieval, a $t$-gate rotation cascade, and the same shift network as~\cite{berti2025efficient}. By Lemma~\ref{lem:bbqram_cost}, the $2k$ BBQRAM queries give $\mathcal{O}(\log_2^2 K)$ query time; the QPU registers are the $k$-qubit address, one $t$-qubit angle register, and the one-qubit target. Since the angles are stored in memory, no reversible adder, divider, square-root, or arcsine subroutine is invoked on the QPU.
\end{proof}
\begin{table}[t]
\centering
\caption{Per-iteration resource comparison: original algorithm of~\cite{berti2025efficient} vs.\ this work.}
\label{tab:comparison}
\scriptsize
\resizebox{\columnwidth}{!}{%
\begin{tabular}{lcc}
\toprule
\textbf{Resource} & \textbf{Original~\cite{berti2025efficient}} & \textbf{This work} \\
\midrule
Fields used in amplitude loop & $2t$ values & $t$ angles \\
Memory Cell width for $\mathcal{R}$ & $1 + 2t$ & $t+1$ \\
Memory Cell width for $\mathcal{C}$ & not supported & $2t$ \\
QPU qubits, magnitude loop & $k+2t+2$ + arith. & $k+t+1$ \\
Gates per iteration & $\tilde{\mathcal{O}}(t) + \mathcal{O}(t)$ & $\mathcal{O}(t)$ ($R_y$ cascade) \\
Reversible arithmetic & add, div, sqrt, arcsin & none \\
Time per iteration & $\mathcal{O}(\log_2 K) + \tilde{\mathcal{O}}(t)$ & $\mathcal{O}(\log_2 K) + \mathcal{O}(t)$ \\
\bottomrule
\end{tabular}
}
\end{table}
\begin{corollary}[Real signed data]\label{cor:real_signed}
For $A\in\mathbb{R}^{M\times N}$, store one additional leaf sign bit $s_z=\mathbf{1}[a_z<0]$ in each memory cell, as illustrated by the orange leaf row of Figure~\ref{fig:angle_tree_theta}. After Theorem~\ref{thm:magnitude_prep}, a BBQRAM query loads $s_z$ in superposition, a controlled $Z$ on the phase target applies $(-1)^{s_z}$, and the reverse query uncomputes the sign register. This prepares the signed real amplitude encoding in $\mathcal{O}(\log_2^2K)$ query time with cell width $t+1$ bits, $2k+2$ BBQRAM queries, and $k+t+2$ QPU qubits.
\end{corollary}

Table~\ref{tab:comparison} compares the amplitude-preparation loop with the previous result~\cite{berti2025efficient}. The primary improvement is the removal of all reversible arithmetic circuits; the angle field also reduces the data loaded by each amplitude-preparation query from two $t$-bit subtree weights to one $t$-bit angle.

The original algorithm invokes a reversible adder~\cite{draper2006logarithmic}, divider, square root, and arcsine~\cite{Sanders_2019_arcsin} in every level. The proposed layout moves these computations to classical preprocessing. The remaining QPU cost is BBQRAM access plus controlled rotations. We count this as query time rather than a complete physical resource estimate.

\section{Extension to Complex-Valued Matrices}\label{sec:complex_extension}
In this section, we extend the state preparation algorithm to complex-valued matrices $A \in \mathbb{C}^{M \times N}$. Section~\ref{sec:complex_memory} describes the memory layout for complex data, which stores the phase angle alongside the precomputed rotation angle in each BBQRAM cell. Section~\ref{sec:two_step_alg} presents the resulting two-step algorithm and analyzes its complexity. The key idea is to reuse the amplitude preparation loop of Section~\ref{sec:eliminating_u2cr} unchanged, and to add Step~2, which encodes the complex phases.

Each entry $a_z \in \mathbb{C}$ is decomposed in polar form as
\begin{equation}\label{eq:polar}
a_z = |a_z|\,e^{i\varphi_z},
\end{equation}
where $|a_z| = \sqrt{\Re(a_z)^2 + \Im(a_z)^2}$ and the phase is
\begin{equation}\label{eq:atan2}
\varphi_z = \atantwo\bigl(\Im(a_z),\;\Re(a_z)\bigr) \;\in\; (-\pi, \pi].
\end{equation}
For $a_z=0$, we set $\varphi_z=0$ by convention; the value is irrelevant because the amplitude of that basis state is zero. The memory stores $\bar\varphi_z=\varphi_z\bmod 2\pi\in[0,2\pi)$ using the phase convention of Section~\ref{sec:cascade}. The segment tree of squared moduli is still built only on $|a_z|^2$, and the corresponding angle tree $\Theta$ of Definition~\ref{def:angle_segment_tree} contains only unsigned magnitude angles; signs and complex phases live exclusively in the leaf-level phase data. The real-valued case is naturally subsumed: for $a_z \geq 0$, $\varphi_z=0$, while for $a_z<0$, $\varphi_z=\pi$. We make the resulting two-tier data structure explicit in Definition~\ref{def:complex_segment_tree}.

\begin{definition}[Complex angle tree]\label{def:complex_segment_tree}
The \emph{complex angle tree} associated with $A \in \mathbb{C}^{M \times N}$ is the pair $\Gamma = (\Theta, \Phi)$, where $\Theta$ is the angle tree of Definition~\ref{def:angle_segment_tree} (built from $\{|a_z|^2\}$), and $\Phi = (\bar\varphi_0, \dots, \bar\varphi_{K-1})$ is a phase layer of $K$ nodes attached one-per-entry below the leaves of $\Theta$. The angle layer $\Theta$ drives the magnitude preparation, while the phase layer $\Phi$ is consumed by a single coherent retrieval at the end of the procedure.
\end{definition}

\subsection{Memory Layout for Complex Data}\label{sec:complex_memory}
\begin{figure}[t]
\centering
\resizebox{\columnwidth}{!}{%
\begin{tikzpicture}[
    >=stealth,
    every node/.style={align=center},
    angbase/.style={
      draw=black!60, line width=0.4pt,
      rounded corners=2pt,
      minimum width=1.2cm, minimum height=0.6cm,
      font=\scriptsize,
      drop shadow={shadow xshift=0.4pt, shadow yshift=-0.4pt, opacity=0.25}},
    anglevzero/.style={angbase, fill=PastelIBMPurple},
    anglevone/.style={angbase, fill=PastelIBMBlue},
    anglevtwo/.style={angbase, fill=PastelIBMPink},
    phase/.style={
      draw=black!60, line width=0.4pt,
      rounded corners=2pt,
      minimum width=0.9cm, minimum height=0.6cm,
      fill=PastelIBMOrange,
      font=\scriptsize,
      drop shadow={shadow xshift=0.4pt, shadow yshift=-0.4pt, opacity=0.25}},
    treeedge/.style={draw=black!50, line width=0.5pt, ->}]
    \node[anglevzero] (r)   at ( 0  , 0)   {$\Theta_{0,0}$};
    \node[anglevone]  (a10) at (-2.4,-1)   {$\Theta_{1,0}$};
    \node[anglevone]  (a11) at ( 2.4,-1)   {$\Theta_{1,1}$};
    \node[anglevtwo]  (a20) at (-3.6,-2)   {$\Theta_{2,0}$};
    \node[anglevtwo]  (a21) at (-1.2,-2)   {$\Theta_{2,1}$};
    \node[anglevtwo]  (a22) at ( 1.2,-2)   {$\Theta_{2,2}$};
    \node[anglevtwo]  (a23) at ( 3.6,-2)   {$\Theta_{2,3}$};
    \draw[treeedge] (r)   -- (a10);  \draw[treeedge] (r)   -- (a11);
    \draw[treeedge] (a10) -- (a20);  \draw[treeedge] (a10) -- (a21);
    \draw[treeedge] (a11) -- (a22);  \draw[treeedge] (a11) -- (a23);
    \node[phase] (p0) at (-4.2,-3.1) {$\bar\varphi_0$};
    \node[phase] (p1) at (-3.0,-3.1) {$\bar\varphi_1$};
    \node[phase] (p2) at (-1.8,-3.1) {$\bar\varphi_2$};
    \node[phase] (p3) at (-0.6,-3.1) {$\bar\varphi_3$};
    \node[phase] (p4) at ( 0.6,-3.1) {$\bar\varphi_4$};
    \node[phase] (p5) at ( 1.8,-3.1) {$\bar\varphi_5$};
    \node[phase] (p6) at ( 3.0,-3.1) {$\bar\varphi_6$};
    \node[phase] (p7) at ( 4.2,-3.1) {$\bar\varphi_7$};
    \draw[treeedge] (a20) -- (p0);  \draw[treeedge] (a20) -- (p1);
    \draw[treeedge] (a21) -- (p2);  \draw[treeedge] (a21) -- (p3);
    \draw[treeedge] (a22) -- (p4);  \draw[treeedge] (a22) -- (p5);
    \draw[treeedge] (a23) -- (p6);  \draw[treeedge] (a23) -- (p7);
\end{tikzpicture}}
\caption{The complex angle tree $\Gamma=(\Theta,\Phi)$ for $K=8$ entries (Definition~\ref{def:complex_segment_tree}). The angle tree $\Theta$ has depth $\log_2 K - 1$ and stores one unsigned amplitude-splitting angle per internal node of the squared-moduli tree $T$; it is consumed level by level during Step~1. The phase layer $\Phi$ holds one entry-wise phase $\bar\varphi_z$ per leaf of $T$ and is consumed by a single coherent retrieval in Step~2.}
\label{fig:gamma_data_structure}
\end{figure}
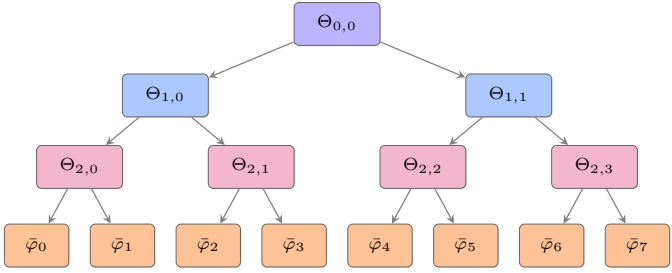
\begin{proposition}[Memory layout for complex matrices]\label{prop:memory_layout_complex}
Let $\Gamma = (\Theta, \Phi)$ be the complex angle tree of Definition~\ref{def:complex_segment_tree}. We map $\Gamma$ into a BBQRAM with $K$ cells $\{\ket{L_z}\}_{z=0}^{K-1}$, each containing $2t$ classical bits loaded coherently into a $2t$-qubit data register:
\begin{equation}
\label{eq:complex_memory_layout}
\ket{L_z}
=
\begin{cases}
\ket{0}^{t}\ket{\bar\varphi_0}^{t}, & \textrm{if } z=0, \\[1ex]
\ket{\Theta_{l(z)-1,d(z)}}^{t}\ket{\bar\varphi_z}^{t}, & \textrm{otherwise,}
\end{cases}
\end{equation}
where $l(z)=\lfloor\log_2 z\rfloor+1$ and $d(z)=z-2^{\lfloor\log_2 z\rfloor}$. Thus each memory cell stores one $t$-bit amplitude angle and one $t$-bit phase. The angle field of $L_0$ is dummy (i.e., set to zero), while its phase field stores the physical phase of entry $a_0$; Step~1 interprets $z\geq1$ as an internal-node index of $\Theta$, and Step~2 interprets every $z$ as a leaf/data index of $\Phi$.

\end{proposition}

We store the unsigned angles of $\Theta$ because all sign and phase information is supplied by $\Phi$. For real-valued matrices, the phase field collapses to a single sign bit, giving the $t+1$ bit layout of Corollary~\ref{cor:real_signed}. Figure~\ref{fig:gamma_data_structure} illustrates the data structure $\Gamma$, separating the angle tree $\Theta$ from the phase layer $\Phi$ attached at its leaves. Figure~\ref{fig:complex_memory_layout} then shows how $\Gamma$ is laid out in the BBQRAM cells.

\begin{figure*}[t]
  \centering
  \resizebox{\textwidth}{!}{%
    \begin{tikzpicture}[
        every label/.append style={font=\large},
        node/.style={
          draw, fill=white, circle, inner sep=0pt, font=\large,
          drop shadow, text width=2.5em, align=center
        },
        leaf/.style={
          shape=rectangle split, rectangle split parts=2,
          rectangle split draw splits, rectangle split horizontal,
          rounded corners, draw, inner sep=3pt,
          minimum width=4em, minimum height=2.5em,
          font=\normalsize, text width=3em, align=center
        },
        leafdummy/.style={
          rectangle split part fill={neutral-left,PastelIBMOrange}
        },
        leaflevzero/.style={
          rectangle split part fill={PastelIBMPurple,PastelIBMOrange}
        },
        leaflevone/.style={
          rectangle split part fill={PastelIBMBlue,PastelIBMOrange}
        },
        leaflevtwo/.style={
          rectangle split part fill={PastelIBMPink,PastelIBMOrange}
        },
      ]
      \def\dx{6.5} \def\dy{1.2} \def\dxii{3.25} \def\leafsep{1.6}

      \node[node] (root) at (0,0) {$\ket{\bullet}$};
      \coordinate (above) at ($(root)+(0,\dy/1.5)$);
      \draw (above) -- (root);
      \node[node] (L)  at ($(root)+(-\dx,-\dy)$) {$\ket{\bullet}$};
      \node[node] (R)  at ($(root)+(\dx,-\dy)$)  {$\ket{\bullet}$};
      \node[node] (LL) at ($(L)+(-\dxii,-\dy)$) {$\ket{\bullet}$};
      \node[node] (LR) at ($(L)+(\dxii,-\dy)$)  {$\ket{\bullet}$};
      \node[node] (RL) at ($(R)+(-\dxii,-\dy)$) {$\ket{\bullet}$};
      \node[node] (RR) at ($(R)+(\dxii,-\dy)$)  {$\ket{\bullet}$};

      \node[leaf, leafdummy, label=below:{\small$\ket{L_0}$}] (LLL) at ($(LL)+(-\leafsep,-\dy)$)
        {$\ket{0}^{t}$\nodepart{two}$\ket{\bar\varphi_0}$};
      \node[leaf, leaflevzero, label=below:{\small$\ket{L_1}$}] (LLR) at ($(LL)+(\leafsep,-\dy)$)
        {$\ket{\theta_1}$\nodepart{two}$\ket{\bar\varphi_1}$};
      \node[leaf, leaflevone, label=below:{\small$\ket{L_2}$}] (LRL) at ($(LR)+(-\leafsep,-\dy)$)
        {$\ket{\theta_2}$\nodepart{two}$\ket{\bar\varphi_2}$};
      \node[leaf, leaflevone, label=below:{\small$\ket{L_3}$}] (LRR) at ($(LR)+(\leafsep,-\dy)$)
        {$\ket{\theta_3}$\nodepart{two}$\ket{\bar\varphi_3}$};
      \node[leaf, leaflevtwo, label=below:{\small$\ket{L_4}$}] (RLL) at ($(RL)+(-\leafsep,-\dy)$)
        {$\ket{\theta_4}$\nodepart{two}$\ket{\bar\varphi_4}$};
      \node[leaf, leaflevtwo, label=below:{\small$\ket{L_5}$}] (RLR) at ($(RL)+(\leafsep,-\dy)$)
        {$\ket{\theta_5}$\nodepart{two}$\ket{\bar\varphi_5}$};
      \node[leaf, leaflevtwo, label=below:{\small$\ket{L_6}$}] (RRL) at ($(RR)+(-\leafsep,-\dy)$)
        {$\ket{\theta_6}$\nodepart{two}$\ket{\bar\varphi_6}$};
      \node[leaf, leaflevtwo, label=below:{\small$\ket{L_7}$}] (RRR) at ($(RR)+(\leafsep,-\dy)$)
        {$\ket{\theta_7}$\nodepart{two}$\ket{\bar\varphi_7}$};

      \draw (root) -- ++(-\dx,0) -- (L);  \draw (root) -- ++(\dx,0)  -- (R);
      \draw (L)  -- ++(-\dxii,0) -- (LL); \draw (L)  -- ++(\dxii,0) -- (LR);
      \draw (R)  -- ++(-\dxii,0) -- (RL); \draw (R)  -- ++(\dxii,0) -- (RR);
      \draw (LL) -- ++(-\leafsep,0) -- (LLL); \draw (LL) -- ++(\leafsep,0) -- (LLR);
      \draw (LR) -- ++(-\leafsep,0) -- (LRL); \draw (LR) -- ++(\leafsep,0) -- (LRR);
      \draw (RL) -- ++(-\leafsep,0) -- (RLL); \draw (RL) -- ++(\leafsep,0) -- (RLR);
      \draw (RR) -- ++(-\leafsep,0) -- (RRL); \draw (RR) -- ++(\leafsep,0) -- (RRR);
    \end{tikzpicture}
  }
  \caption{BBQRAM mapping of the complex angle tree $\Gamma=(\Theta,\Phi)$ (Proposition~\ref{prop:memory_layout_complex}). Cell $L_0$ holds a dummy zero in the angle field; for $z\geq 1$, each cell stores one node of $\Theta$ as the rotation-angle field $\ket{\theta_z}$ and the corresponding leaf phase $\ket{\bar\varphi_z}$ from $\Phi$, for a total width of $2t$ bits. The angle field drives magnitude splitting through cascades of controlled $R_y$ gates; the phase field is applied through cascades of controlled phase gates.}
  \label{fig:complex_memory_layout}
\end{figure*}
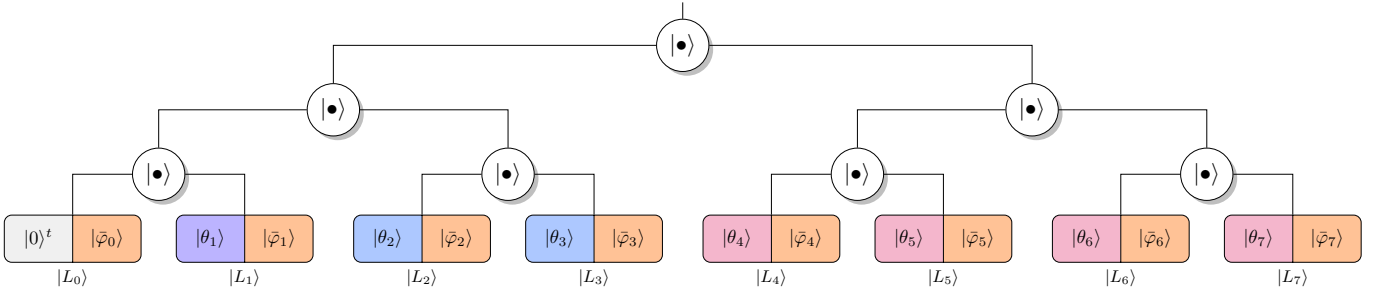

\subsection{Two-Step Algorithm}\label{sec:two_step_alg}
The algorithm operates in two steps: \emph{Step~1} prepares the moduli using Theorem~\ref{thm:magnitude_prep}, and \emph{Step~2} applies the complex phases via one additional BBQRAM query and a cascade of controlled $P$ gates. Algorithm~\ref{alg:complex} provides the pseudocode, and Theorem~\ref{thm:state_prep_complex} states the formal guarantee.

\begin{algorithm}[t]
\caption{State preparation for complex matrices}\label{alg:complex}
\KwIn{$A \in \mathbb{C}^{M \times N}$, $K=MN$, $k=\log_2 K$, precision $t$}
\KwOut{$\frac{1}{\|A\|_F}\sum_{i,j} a_{i,j}\ket{i}\ket{j}$}
\BlankLine
\tcc{Classical preprocessing}
\For{$z = 0$ \KwTo $K-1$}{
  $|a_z|^2 \leftarrow \Re(a_z)^2 + \Im(a_z)^2$\;
  $\bar\varphi_z \leftarrow 0$ if $a_z=0$, otherwise $\atantwo(\Im(a_z),\,\Re(a_z)) \bmod 2\pi$\;
}
Build segment tree $T$ on $\{|a_z|^2\}$\;
Write $\ket{\theta_0=0}\ket{\bar\varphi_0}$ into BBQRAM cell $0$\;
\For{$z = 1$ \KwTo $K-1$}{
  $\theta_z \leftarrow 0$ if $T_L(z)+T_R(z)=0$, otherwise $2\arcsin\!\bigl(\sqrt{\frac{T_R(z)}{T_L(z)+T_R(z)}}\bigr)$\;
  Write $\ket{\theta_z}\ket{\bar\varphi_z}$ into BBQRAM memory cell $z$\;
}
\BlankLine
\tcc{Step 1: Amplitude preparation}
Initialize $\ket{0}^t_{\mathrm{w}_\theta}\ket{0}^t_{\mathrm{w}_\varphi}\ket{0}_{\mathrm{v}}\ket{0\cdots01}_{\mathrm{a}}$\;
\For{$h = 1$ \KwTo $k$}{
  Retrieve $\ket{\theta_z}\ket{\bar\varphi_z}$ from BBQRAM into $\mathrm{w}_\theta,\mathrm{w}_\varphi$ \tcp*{$\mathcal{O}(\log_2 K)$}
  Apply cascade $\mathrm{C}_{\mathrm{w}_\theta}R_{y\mapsto\mathrm{v}}(\theta_z)$ \tcp*{$\mathcal{O}(t)$}
  Uncompute $\mathrm{w}_\theta,\mathrm{w}_\varphi$ \tcp*{$\mathcal{O}(\log_2 K)$}
  Left circular shift on $\mathrm{v},\mathrm{a}$ (network of~\cite{berti2025efficient})
}
\BlankLine
\tcc{Step 2: Phase encoding}
Retrieve $\ket{\theta_z}\ket{\bar\varphi_z}$ from BBQRAM into $\mathrm{w}_\theta,\mathrm{w}_\varphi$ \tcp*{$\mathcal{O}(\log_2 K)$}
Apply cascade $\mathrm{C}_{\mathrm{w}_\varphi}P_{\mapsto\mathrm{v}}(\bar\varphi_z)$ \tcp*{$\mathcal{O}(t)$}
Uncompute $\mathrm{w}_\theta,\mathrm{w}_\varphi$ \tcp*{$\mathcal{O}(\log_2 K)$}
\Return state on register $\mathrm{a}$\;
\end{algorithm}
\begin{theorem}[Efficient state preparation for complex matrices]\label{thm:state_prep_complex}
Given $A \in \mathbb{C}^{M \times N}$ with $M = 2^m$, $N = 2^n$, and $K = MN$, and a complex angle tree $\Gamma=(\Theta,\Phi)$ embedded in BBQRAM according to Proposition~\ref{prop:memory_layout_complex}, there exists a unitary procedure $E_A$ whose action on the initialized QPU registers is
\[
\begin{aligned}
E_A:\;&\ket{0}^t_{\mathrm{w}_\theta}\ket{0}^t_{\mathrm{w}_\varphi}
\ket{0}_{\mathrm{v}}\ket{0^{k-1}1}_{\mathrm{a}}\\
&\longmapsto
\ket{0}^t_{\mathrm{w}_\theta}\ket{0}^t_{\mathrm{w}_\varphi}\ket{1}_{\mathrm{v}}
\frac{1}{\|A\|_F}\sum_{z=0}^{K-1} a_z\ket{z}_{\mathrm{a}}^k .
\end{aligned}
\]
up to fixed-point. Equivalently, the address register contains $\frac{1}{\|A\|_F}\sum_{i,j}a_{i,j}\ket{i}^m\ket{j}^n$ after discarding the deterministic work registers. The BBQRAM query time is $\mathcal{O}(\log_2^2 K)$, using $k + 2t + 1$ qubits on the QPU.
\end{theorem}

The procedure operates on four registers: a $t$-qubit register $\ket{0}^t_{\mathrm{w}_\theta}$ for the rotation angle, a $t$-qubit register $\ket{0}^t_{\mathrm{w}_\varphi}$ for the phase angle, a $1$-qubit register $\ket{0}^1_{\mathrm{v}}$ for amplitude encoding, and a $k$-qubit address register $\ket{0}^k_{\mathrm{a}}$ that also stores the final output state.

\begin{proof}
Classical preprocessing computes the squared moduli, the fixed-point phases $\bar\varphi_z$, and the angle tree $\Theta$, then writes the cells of Proposition~\ref{prop:memory_layout_complex}. Step~1 is Theorem~\ref{thm:magnitude_prep} applied to the angle field; although the full complex cell is loaded and uncomputed, only $\mathrm{w}_\theta$ is used in the rotation cascade, and Lemma~\ref{lem:routing_marker} leaves $\mathrm{v}=\ket{1}$ with the address holding the leaf index. Step~2 retrieves the phase field at that leaf address and applies the controlled-$P$ cascade, which multiplies the branch by $e^{i\bar\varphi_z}=e^{i\varphi_z}$ because $\mathrm{v}=\ket{1}$. The reverse retrieval cleans the work registers, giving amplitudes $|a_z|e^{i\varphi_z}/\|A\|_F=a_z/\|A\|_F$.

The query count is $2k$ for the magnitude loop plus one phase query and its reverse. Hence the BBQRAM query time is $\mathcal{O}(\log_2^2 K)$, with the Step~2 $\mathcal{O}(\log_2 K)$ contribution lower order. The QPU registers are the $k$-qubit address, two $t$-qubit work registers, and the one-qubit target.
\end{proof}

The effect of fixed-point storage and rotation synthesis can be bounded by a standard
hybrid argument over the $k$ magnitude layers and the final phase layer. If the stored
angles satisfy $|\widetilde\theta_z-\theta_z|\leq\delta_\theta$, the stored phases satisfy
$d_{2\pi}(\widetilde\varphi_z,\varphi_z)\leq\delta_\varphi$, and the implemented
$R_y$ and phase cascades have operator-norm errors at most $\epsilon_y$ and
$\epsilon_\varphi$, respectively, then
\begin{equation}\label{eq:prec_corr}
\|\ket{\widetilde A}-\ket{A}\|
\lesssim
k\left(\frac{\delta_\theta}{2}+\epsilon_y\right)
+\delta_\varphi+\epsilon_\varphi .
\end{equation}
For nearest-grid rounding with spacings
$\Delta_\theta=2^{2-t}$ and $\Delta_\varphi=2\pi/2^t$, one has
$\delta_\theta\leq 2^{1-t}$ and $\delta_\varphi\leq \pi 2^{-t}$. Thus choosing
$t=\mathcal{O}(\log((k+1)/\eta))$ and synthesizing each cascade to error
$\mathcal{O}(\eta/(k+1))$ suffices to achieve target state error $\eta$.

Table~\ref{tab:complex_comparison} provides a comprehensive resource comparison.

\begin{table}[t]
\centering
\caption{Resource comparison separating the original real-valued method, the real-valued specialization of this work, and the complex-valued construction. ``QPU qubits'' excludes BBQRAM switches and memory cells. Rotation-synthesis costs are separate from the BBQRAM query count. The original method's $2k{+}1$ count corresponds to $2k$ queries for the magnitude loop plus a single phase-kickback sign step that requires no reverse query, whereas the new methods explicitly perform a final query and its reverse for the sign/phase field, giving $2k{+}2$.}
\label{tab:complex_comparison}
\scriptsize
\resizebox{\columnwidth}{!}{%
\begin{tabular}{lccc}
\toprule
\textbf{Resource} & \textbf{Original real~\cite{berti2025efficient}} & \textbf{This work, real} & \textbf{This work, complex} \\
\midrule
Input domain & $\mathbb{R}^{M \times N}$ & $\mathbb{R}^{M \times N}$ & $\mathbb{C}^{M \times N}$ \\
BBQRAM query time & $\mathcal{O}(\log_2^2 K)$ & $\mathcal{O}(\log_2^2 K)$ & $\mathcal{O}(\log_2^2 K)$ \\
QPU qubits & $k + 2t + 2$ + arith. & $k+t+2$ & $k + 2t + 1$ \\
QPU qubits at error $\eta$ & $k+2t_\eta+2$ + arith. & $k+t_\eta+2$ & $k+2t_\eta+1$ \\
Memory Cell width & $1 + 2t$ bits & $t+1$ bits & $2t$ bits \\
BBQRAM queries & $2k + 1$ & $2k + 2$ & $2k + 2$ \\
Rev.\ arithmetic & yes ($U_{2\mathrm{CR}}$) & none & none \\
Online rotations & $R_y$ cascade + arith. & $R_y$ cascade + $CZ$ & $R_y$ and $P$ cascades \\
Classical precomp. & $\mathcal{O}(K)$ & $\mathcal{O}(K\poly(t))$ & $\mathcal{O}(K\poly(t))$ \\
\bottomrule
\end{tabular}
}
\end{table}

The real-valued specialization adds a final query and its reverse to load the sign bit after magnitude preparation; in exchange, the arithmetic circuits and one of the two $t$-bit magnitude fields are removed. The extension to complex matrices replaces the sign bit by a $t$-bit phase field, adding $t$ QPU qubits for $\mathrm{w}_\varphi$ and $t$ additional bits per BBQRAM cell relative to the real specialization. For target state error $\eta$, the row with $t_\eta=\mathcal{O}(\log_2((k+1)/\eta))$ follows from 
Equation~\eqref{eq:prec_corr}, 
excluding synthesis ancillas. The Step~2 query accesses all $K$ memory cells in full superposition, whereas each Step~1 iteration accesses one segment-tree level; in the hardware-BBQRAM model, both are $\mathcal{O}(\log_2 K)$-time queries by Lemma~\ref{lem:bbqram_cost}. The query table does not include the fault-tolerant synthesis cost of the controlled rotations; Equation~\ref{eq:prec_corr} 
states the corresponding precision dependence.

\section{Conclusion and Future Work}\label{sec:conclusion}
This work showed that the BBQRAM-based quantum state preparation framework of~\cite{berti2025efficient} admits two concrete improvements. First, classical precomputation of segment-tree rotation angles eliminates the reversible arithmetic inside $U_{2\mathrm{CR}}$ from the magnitude-preparation loop. Second, a two-step magnitude-then-phase procedure extends the same architecture-aware framework to complex-valued matrices. The real signed case is handled as a one-bit phase specialization, avoiding any need to assign signs to internal segment-tree subtrees. Both improvements preserve the $\mathcal{O}(\log_2^2(MN))$ BBQRAM query-time bound.

We observe that the $\mathcal{O}(\log_2^2 K)$ time bound is contingent on the hardware-based BBQRAM model with pipelined routing. In~\cite{Jaques2025qramsurveycritique}, the authors show that any \emph{circuit}-level BBQRAM implementation faces fundamental lower bounds --- $\Omega(2^n)$ logical gates and $\Omega(\sqrt{2^n})$ non-Clifford gates --- that preclude polylogarithmic gate counts in the circuit model. Building a physical BBQRAM device and interfacing it with a generic quantum processor thus remains an open engineering challenge, although recent works report concrete progress toward practical implementations~\cite{Hann_2019, cesa2025, dalzell2025, Wang_2025, PRXQuantum.5.020312, PRXQuantum.2.030319, deriso2025}.

Within this hardware model, our two contributions have complementary significance. Eliminating $U_{2\mathrm{CR}}$ is primarily a practical improvement: the asymptotic query complexity is unchanged, but the quantum circuit no longer contains reversible adder, divider, square-root, and arcsine subroutines. The complex extension is primarily a modeling contribution: it gives an explicit phase-loading layer for the same architecture, extending the framework's applicability to domains in which complex data arises naturally.

The intended downstream regime is repeated preparation of the same data, where the $\mathcal{O}(K\poly(t))$ preprocessing and BBQRAM-writing cost can be amortized across many calls.  Concerning scaling to realistic problem sizes, the per-iteration QPU structure --- one BBQRAM retrieval followed by a cascade of controlled rotations --- is independent of $K$ and does not change across Table~\ref{tab:scaling}. The table should be read as a QPU-register and BBQRAM-memory accounting exercise, not as a full feasibility estimate: it excludes physical BBQRAM switches, error correction, control hardware, magic-state factories, and rotation-synthesis cost.

\begin{table}[t]
\centering
\caption{Resource counts for complex-valued amplitude encoding of $K$ amplitudes at precision $t = 32$ bits per angle (Theorem~\ref{thm:state_prep_complex}): QPU qubits $= k+2t+1$, BBQRAM memory $= 2tK$ classical bits, queries $= 2k+2$, query time $\mathcal{O}(\log_2^2 K)$ (Lemma~\ref{lem:bbqram_cost}).}
\label{tab:scaling}
\small
\begin{tabular}{rrrrr}
\toprule
\textbf{$K$} & \textbf{$k$} & \textbf{QPU qubits} & \textbf{BBQRAM (bits)} & \textbf{Queries}\\
\midrule
$2^{10}$ & $10$ & $75$ & $\approx 6.6{\times}10^4$ & $22$ \\
$2^{20}$ & $20$ & $85$ & $\approx 6.7{\times}10^7$ & $42$ \\
$2^{30}$ & $30$ & $95$ & $\approx 6.9{\times}10^{10}$ & $62$ \\
\bottomrule
\end{tabular}
\end{table}
Several directions for future work are worth exploring. First, the present algorithm assumes classical data: changing one entry requires updating the corresponding phase, all angles along the affected root-to-leaf path, and the BBQRAM contents. We do not provide a coherent dynamic-update primitive. Second, specialized memory layouts for sparse matrices could shrink the BBQRAM size from $\mathcal{O}(MN)$ to $\mathcal{O}(\mathrm{nnz}(A))$ by storing only non-zero entries, and would require a modified address-space mapping and a segment-tree construction that handles variable-length branches. Third, adaptive-precision schemes that allow the bit width $t$ to vary per entry could reduce memory requirements while preserving overall accuracy, at the price of a non-uniform cascade structure. Fourth, this work prepares an amplitude-encoded state; it does not by itself provide row/column access or a block encoding of $A$. Extending the framework to block encodings~\cite{sunderhauf2024block} would broaden its applicability to quantum singular-value-transformation-based algorithms~\cite{gilyen2019quantum}. Eventually, a full fault-tolerant resource estimate should combine the query model used here with a concrete physical BBQRAM architecture and rotation-synthesis costs. As quantum hardware matures, architecture-aware approaches to state preparation --- such as the one presented in this work --- offer a complementary perspective to circuit-based methods and can inform the co-design of quantum software and memory architectures.

\section{Numerical Example}\label{sec:numerical_example}
In this section, we illustrate the algorithm on the complex-valued matrix
\[
A \in \mathbb{C}^{2 \times 4} = \begin{bmatrix}
2+i & -1+2i & 3 & -i \\
1-i & 2i & -2+i & 1+i
\end{bmatrix},
\]
for which $M = 2$, $N = 4$, $K = MN = 8$, $k = \log_2 K = 3$, and $\|A\|_F = \sqrt{33}$. We trace the classical preprocessing, Step~1 (amplitude preparation), and Step~2 (phase encoding) step by step.

\paragraph{Classical preprocessing}
The squared moduli of the entries are $[5, 5, 9, 1, 2, 4, 5, 2]$. The precomputed angles $\theta_z$ derived from the sibling sums populate $\Theta$, and the phases $\bar\varphi_z=\atantwo(\Im(a_z),\Re(a_z))$ populate the leaf layer $\Phi$; both are tabulated jointly in the BBQRAM-cell layout of Figure~\ref{fig:bbqram_layout_example}.

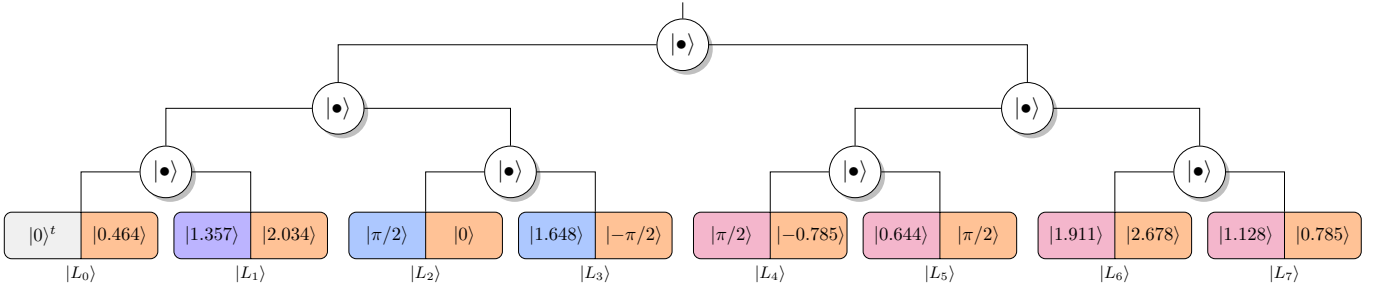
\begin{figure*}[ht]
  \centering
  \resizebox{\textwidth}{!}{%
    \begin{tikzpicture}[
        every label/.append style={font=\large},
        node/.style={
          draw, fill=white, circle, inner sep=0pt, font=\large,
          drop shadow, text width=2.5em, align=center
        },
        leaf/.style={
          shape=rectangle split, rectangle split parts=2,
          rectangle split draw splits, rectangle split horizontal,
          rounded corners, draw, inner sep=3pt,
          minimum width=4.75em, minimum height=2.5em,
          font=\normalsize, text width=3.5em, align=center
        },
        leafdummy/.style={
          rectangle split part fill={neutral-left,PastelIBMOrange}
        },
        leaflevzero/.style={
          rectangle split part fill={PastelIBMPurple,PastelIBMOrange}
        },
        leaflevone/.style={
          rectangle split part fill={PastelIBMBlue,PastelIBMOrange}
        },
        leaflevtwo/.style={
          rectangle split part fill={PastelIBMPink,PastelIBMOrange}
        },
      ]
      \def\dx{6.5} \def\dy{1.2} \def\dxii{3.25} \def\leafsep{1.6}

      \node[node] (root) at (0,0) {$\ket{\bullet}$};
      \coordinate (above) at ($(root)+(0,\dy/1.5)$);
      \draw (above) -- (root);
      \node[node] (L)  at ($(root)+(-\dx,-\dy)$) {$\ket{\bullet}$};
      \node[node] (R)  at ($(root)+(\dx,-\dy)$)  {$\ket{\bullet}$};
      \node[node] (LL) at ($(L)+(-\dxii,-\dy)$) {$\ket{\bullet}$};
      \node[node] (LR) at ($(L)+(\dxii,-\dy)$)  {$\ket{\bullet}$};
      \node[node] (RL) at ($(R)+(-\dxii,-\dy)$) {$\ket{\bullet}$};
      \node[node] (RR) at ($(R)+(\dxii,-\dy)$)  {$\ket{\bullet}$};

      \node[leaf, leafdummy, label=below:{\small$\ket{L_0}$}] (LLL) at ($(LL)+(-\leafsep,-\dy)$)
        {$\ket{0}^{t}$\nodepart{two}$\ket{0.464}$};
      \node[leaf, leaflevzero, label=below:{\small$\ket{L_1}$}] (LLR) at ($(LL)+(\leafsep,-\dy)$)
        {$\ket{1.357}$\nodepart{two}$\ket{2.034}$};
      \node[leaf, leaflevone, label=below:{\small$\ket{L_2}$}] (LRL) at ($(LR)+(-\leafsep,-\dy)$)
        {$\ket{\pi/2}$\nodepart{two}$\ket{0}$};
      \node[leaf, leaflevone, label=below:{\small$\ket{L_3}$}] (LRR) at ($(LR)+(\leafsep,-\dy)$)
        {$\ket{1.648}$\nodepart{two}$\ket{-\pi/2}$};
      \node[leaf, leaflevtwo, label=below:{\small$\ket{L_4}$}] (RLL) at ($(RL)+(-\leafsep,-\dy)$)
        {$\ket{\pi/2}$\nodepart{two}$\ket{-0.785}$};
      \node[leaf, leaflevtwo, label=below:{\small$\ket{L_5}$}] (RLR) at ($(RL)+(\leafsep,-\dy)$)
        {$\ket{0.644}$\nodepart{two}$\ket{\pi/2}$};
      \node[leaf, leaflevtwo, label=below:{\small$\ket{L_6}$}] (RRL) at ($(RR)+(-\leafsep,-\dy)$)
        {$\ket{1.911}$\nodepart{two}$\ket{2.678}$};
      \node[leaf, leaflevtwo, label=below:{\small$\ket{L_7}$}] (RRR) at ($(RR)+(\leafsep,-\dy)$)
        {$\ket{1.128}$\nodepart{two}$\ket{0.785}$};

      \draw (root) -- ++(-\dx,0) -- (L);  \draw (root) -- ++(\dx,0)  -- (R);
      \draw (L)  -- ++(-\dxii,0) -- (LL); \draw (L)  -- ++(\dxii,0) -- (LR);
      \draw (R)  -- ++(-\dxii,0) -- (RL); \draw (R)  -- ++(\dxii,0) -- (RR);
      \draw (LL) -- ++(-\leafsep,0) -- (LLL); \draw (LL) -- ++(\leafsep,0) -- (LLR);
      \draw (LR) -- ++(-\leafsep,0) -- (LRL); \draw (LR) -- ++(\leafsep,0) -- (LRR);
      \draw (RL) -- ++(-\leafsep,0) -- (RLL); \draw (RL) -- ++(\leafsep,0) -- (RLR);
      \draw (RR) -- ++(-\leafsep,0) -- (RRL); \draw (RR) -- ++(\leafsep,0) -- (RRR);
    \end{tikzpicture}
  }
  \caption{BBQRAM-cell layout for the example matrix, instantiating Figure~\ref{fig:complex_memory_layout} with the precomputed angles $\theta_z$ and phases $\bar\varphi_z$ derived above. Cell $\ket{L_0}$ stores a dummy zero in the angle field and the phase $\bar\varphi_0\approx 0.464$; for $z\geq 1$, each cell holds the angle $\theta_z=\Theta_{l(z)-1,d(z)}$ alongside the phase $\bar\varphi_z$.}
  \label{fig:bbqram_layout_example}
\end{figure*}
The precomputed rotation angles $\theta_z$ (via Definition~\ref{def:theta_precompute}) and the phase angles $\varphi_z$ (via $\atantwo$) are tabulated together with the BBQRAM memory cell contents later in this section.

\paragraph{Construction of $\Theta$ and $\Phi$}
The angle tree $\Theta$ has depth $\log_2 K - 1 = 2$ and $K-1 = 7$ nodes. The seven precomputed angles $\theta_1,\dots,\theta_7$ are obtained from Definition~\ref{def:theta_precompute} via the sibling-sum ratios of $T$.
The phase layer $\Phi=(\bar\varphi_0,\dots,\bar\varphi_7)$ is obtained directly from the entries via $\bar\varphi_z = \atantwo(\Im(a_z),\Re(a_z))\bmod 2\pi$. Figure~\ref{fig:bbqram_layout_example} shows $\Theta$ and $\Phi$ jointly in the BBQRAM-cell layout for this example: each cell stores $2t$ bits, namely the angle field $\theta_z=\Theta_{l(z)-1,d(z)}$ (with $\theta_0=0$ a dummy) and the phase field $\bar\varphi_z$ (signed representatives in $(-\pi,\pi]$; the BBQRAM stores the equivalent values modulo $2\pi$).


\paragraph{Step~1: Amplitude preparation}
The initial state is $\ket{0}_{\mathrm{w}_\theta}^t\ket{0}_{\mathrm{w}_\varphi}^t\ket{0}_{\mathrm{v}}\ket{001}_{\mathrm{a}}$, with the address register pointing to memory cell $\ket{L_1}$. We trace through the three iterations. For conciseness, in what follows we display only the address register; the work registers $\mathrm{w}_\theta,\mathrm{w}_\varphi,\mathrm{v}$ start in $\ket{0}$, are uncomputed at the end of each query, and are discarded at the conclusion of the procedure (Theorem~\ref{thm:state_prep_complex}).

\emph{Iteration $h=1$:} the address $\ket{001}$ accesses memory cell $\ket{L_1}$, loading $\theta_1 \approx 1.357$ into $\mathrm{w}_\theta$. The cascade $R_y(\theta_1)$ splits the target qubit with $\cos(\theta_1/2) = \sqrt{20/33}$ and $\sin(\theta_1/2) = \sqrt{13/33}$, correctly assigning probability $20/33$ to the left subtree ($T_{1,0} = 20$) and $13/33$ to the right subtree ($T_{1,1} = 13$). After uncomputing $\mathrm{w}_\theta$ and performing the left circular shift, $\mathrm{v}$ moves into $\mathrm{a}$ and a fresh $\ket{0}$ enters $\mathrm{v}$, yielding
\[
\frac{1}{\sqrt{33}}\Bigl(\sqrt{20}\,\ket{010}_{\mathrm{a}} + \sqrt{13}\,\ket{011}_{\mathrm{a}}\Bigr).
\]
The addresses $\ket{010}$ and $\ket{011}$ now point to cells $\ket{L_2}$ and $\ket{L_3}$ at the next level.

\emph{Iteration $h=2$:} the address register is in superposition over $\ket{010}$ and $\ket{011}$, and the BBQRAM routes these two addresses simultaneously. This loads $\theta_2 = \pi/2$ (for address $010$) and $\theta_3 \approx 1.648$ (for address $011$) into $\mathrm{w}_\theta$ in superposition. After the cascade and uncomputation, the $\sqrt{20}$ branch splits into $\sqrt{10}\ket{0} + \sqrt{10}\ket{1}$ (equal splitting since $T_{2,0} = T_{2,1} = 10$), while the $\sqrt{13}$ branch splits into $\sqrt{6}\ket{0} + \sqrt{7}\ket{1}$ (since $T_{2,2} = 6$ and $T_{2,3} = 7$). After the circular shift, the state becomes $$\frac{1}{\sqrt{33}}(\sqrt{10}\,\ket{100} + \sqrt{10}\,\ket{101} + \sqrt{6}\,\ket{110} + \sqrt{7}\,\ket{111}).$$

\emph{Iteration $h=3$:} the address register accesses cells $\ket{L_4}$--$\ket{L_7}$ in a four-way superposition. After the cascade, uncomputation, and final circular shift, the state becomes
\begin{multline*}
\frac{1}{\sqrt{33}}\bigl(\sqrt{5}\ket{000} + \sqrt{5}\ket{001} + 3\ket{010} + \ket{011}+\\
 + \sqrt{2}\ket{100} + 2\ket{101} + \sqrt{5}\ket{110} + \sqrt{2}\ket{111}\bigr),
\end{multline*}
which encodes the moduli $|a_z|$ as amplitudes, with $\mathrm{v} = \ket{1}$ for all terms. We verify that $(\sqrt{5})^2 + (\sqrt{5})^2 + 9 + 1 + 2 + 4 + 5 + 2 = 33 = \|A\|_F^2$, confirming normalization.

\paragraph{Step~2: Phase encoding}
A single BBQRAM query loads the eight phase angles into $\mathrm{w}_\varphi$ and the cascade of controlled $P$ gates applies $e^{i\varphi_z}$ to each term via $P(\varphi_z)\ket{1}_{\mathrm{v}} = e^{i\varphi_z}\ket{1}_{\mathrm{v}}$. Multiplying each modulus from Step~1 by its phase factor gives the polar form
\begin{multline*}
\frac{1}{\sqrt{33}}\bigl(\sqrt{5}\,e^{i\cdot 0.464}\ket{000} + \sqrt{5}\,e^{i\cdot 2.034}\ket{001}+\\
+ 3\ket{010} + e^{-i\pi/2}\ket{011} + \sqrt{2}\,e^{-i\pi/4}\ket{100}+\\
+ 2\,e^{i\pi/2}\ket{101} + \sqrt{5}\,e^{i\cdot 2.678}\ket{110} + \sqrt{2}\,e^{i\pi/4}\ket{111}\bigr),
\end{multline*}
which coincides term-by-term via $|a_z|\,e^{i\varphi_z} = a_z$, and re-indexing the leaf address $z\in\{0,\dots,7\}$ as $z=iN+j$ with $i\in\{0,1\}$, $j\in\{0,1,2,3\}$, yields
\begin{multline*}
\frac{1}{\sqrt{33}}\bigl((2+i)\ket{0}\ket{0} + (-1+2i)\ket{0}\ket{1} + 3\ket{0}\ket{2} - i\,\ket{0}\ket{3}+\\
+ (1-i)\ket{1}\ket{0} + 2i\,\ket{1}\ket{1} + (-2+i)\ket{1}\ket{2} + (1+i)\ket{1}\ket{3}\bigr).
\end{multline*}

\section*{Acknowledgment}
This work is supported by the National Centre on HPC, Big Data and Quantum Computing - SPOKE 10 (Quantum Computing) and received funding from the European Union Next-GenerationEU - National Recovery and Resilience Plan (NRRP) -- MISSION \emph{4} COMPONENT \emph{2}, INVESTMENT N.\emph{1.4} -- CUP N. \emph{I53C22000690001}. A. Berti was supported by INdAM - GNCS Project N. \emph{E53C24001950001}. F. Ghisoni was supported by the `National Quantum Science Technology Institute' (NQSTI, PE4) within the PNRR project \emph{PE0000023}.

\bibliographystyle{IEEEtran}
\bibliography{main}

\end{document}